\documentclass[oneside,english]{amsart}
\usepackage[T1]{fontenc}
\usepackage[latin9]{inputenc}
\usepackage{amsthm}
\usepackage{amssymb}
\usepackage{graphicx}
\usepackage{esint}
\usepackage[authoryear]{natbib}

\makeatletter
\numberwithin{equation}{section}
\numberwithin{figure}{section}
\theoremstyle{plain}
\newtheorem{thm}{\protect\theoremname}
  \theoremstyle{definition}
  \newtheorem{defn}[thm]{\protect\definitionname}
  \theoremstyle{remark}
  \newtheorem{rem}[thm]{\protect\remarkname}
  \theoremstyle{plain}
  \newtheorem{prop}[thm]{\protect\propositionname}

\makeatother

\usepackage{babel}
  \providecommand{\definitionname}{Definition}
  \providecommand{\propositionname}{Proposition}
  \providecommand{\remarkname}{Remark}
\providecommand{\theoremname}{Theorem}

\begin{document}

\title[double-slit with negative probabilities]{Negative probabilities and Counterfactual Reasoning on the double-slit
Experiment }

\author{J. Acacio de Barros$^{*}$, Gary Oas$^{\dagger}$, and Patrick Suppes$^{\star}$ }

\address{$^{*}$ Liberal Studies Program, San Francisco State University,
San Francisco, CA, 94132.\\
$^{\dagger}$Stanford Pre-Collegiate Studies, Stanford University,Stanford,
CA, 94305-4101.\\
$^{\star}$ CSLI, Ventura Hall, Stanford University, Stanford, CA,
94305-4115. }
\begin{abstract}
In this paper we attempt to establish a theory of negative (quasi)
probability distributions from fundamental principles and apply it
to the study of the double-slit experiment in quantum mechanics. We
do so in a way that preserves the main conceptual issues intact but
allow for a clearer analysis, by representing the double-slit experiment
in terms of the Mach-Zehnder interferometer, and show that the main
features of quantum systems relevant to the double-slit are present
also in the Mach-Zehnder. This converts the problem from a continuous
to a discrete random variable representation. We then show that, for
the Mach-Zehnder interferometer, negative probabilities do not exist
that are consistent with interference and which-path information,
contrary to what Feynman believed. However, consistent with Scully
et al.'s experiment, if we reduce the amount of experimental information
about the system and rely on counterfactual reasoning, a joint negative
probability distribution can be constructed for the Mach-Zehnder experiment. 
\end{abstract}

\date{\today. Early draft. Do not quote without authors' permission. }

\maketitle
Two of the authors (JdB, GO) would like to express their gratitude
for having the honor to contribute to this volume recognizing Patrick
Suppes. It is with great sadness that Pat\textquoteright s passing
came so soon. He did not have a chance to contribute to this final
version, but the core ideas put forth here stem from him, and we believe
that he would be pleased with the final result. However, we emphasize
that any errors are the exclusive responsibility of the first two
authors, and would not be present if the paper had gone through Pat\textquoteright s
usual rigorous review. We also want to take this opportunity to express
our indebtedness to Pat for his guidance and patience over the past
few decades. Pat introduced both of us to the importance of joint
probability distributions in quantum mechanics, and was to us not
only a collaborator and mentor, but also a friend, and we heartily
dedicate this paper to him.

\section{Introduction}

Ever since its inception, quantum mechanics has not ceased to perplex
physicists with its counter-intuitive descriptions of nature. For
instance, in their famous paper Albert Einstein, Boris Podolsky, and
Nathan Rosen (EPR) argued the incompleteness of the quantum mechanical
description \citep{einstein_can_1935}. At the core of their argument
was the superposition of two wavefunctions where properties of two
particles far apart, A and B, were highly correlated. Since both particles
are spatially separated, EPR argued that a measurement on A should
not affect B. Therefore, we could use the correlation and a measurement
on A to infer the value of a property in B \emph{without disturbing
}it. Thus, concluded EPR, the quantum mechanical description of nature
had to be incomplete, as it did not allow the values of a property
to be fixed before an experiment was performed.

In 1964, John Bell showed that not only quantum mechanics was incomplete,
but also that a complete description of physical reality such as the
one espoused by EPR was incompatible with quantum mechanical predictions
\citep{bell_einstein-podolsky-rosen_1964}. Later on, Alain Aspect,
Jean Dalibard, and Gérard Roger, in an impressive and technically
challenging experiment, obtained correlation measurements between
measurement events separated by a spacelike interval. Their correlations
supported quantum mechanics, in disagreement with EPR's metaphysical
views \citep{aspect_experimental_1981,aspect_experimental_1982}.
Other puzzling results followed, like the Kochen-Specker theorem \citep{kochen_problem_1967,kochen_problem_1975},
Wheeler's delayed choice experiment \citep{wheeler_past_1978,jacques_experimental_2007},
the quantum eraser \citep{scully_quantum_1982}, and the Greenberger-Horne-Zeilinger
paradox \citep{greenberger_going_1989,de_barros_inequalities_2000,de_barros_probabilistic_2001},
to mention a few. 

What most of the above examples have in common (with the exception
of Kochen-Specker) is that they all use superpositions of quantum
states. There is nothing more puzzling in quantum mechanics than the
fact that a given system can be in a state with two incompatible properties
simultaneously ``present.'' For example, let $\hat{O}$ be an observable
corresponding to a property $\mathbf{O}$, and $|o_{1}\rangle$ and
$|o_{2}\rangle$ two eigenstates of $\hat{O}$ with eigenvalues $o_{1}$
and $o_{2}$. If the quantum system is in the state $|o_{1}\rangle$,
we may say that it has an objective property $\mathbf{O}$ and its
value is $o_{1}$, in the sense that if we measure this system for
property $\mathbf{O}$, the outcome of this measurement will be $o_{1}$
with probability one. Similarly if the system is in the state $|o_{2}\rangle$.
However, it is also in principle possible to prepare a system in a
quantum state that is the superposition of $|o_{1}\rangle$ and $|o_{2}\rangle$,
i.e. $c_{1}|o_{1}\rangle+c_{2}|o_{2}\rangle$, where $c_{1}$ and
$c_{2}$ are any complex numbers satisfying the constraint $\left|c_{1}\right|^{2}+\left|c_{2}\right|^{2}=1$.
At a first glance, this may not seem puzzling, as it is just telling
us that the system is perhaps in a state where the value of $\mathbf{O}$
is unknown, except that it can either be $o_{1}$ ($o_{2}$) with
probability $\left|c_{1}\right|^{2}$ ($\left|c_{2}\right|^{2}$).
The perplexing aspects of superpositions come from the study of properties
such as $\mathbf{O}$ taken in conjunction with other properties,
say $\mathbf{O}'$, in cases where their corresponding observables,
$\hat{O}$ and $\hat{O}'$, do not commute ($[\hat{O},\hat{O}']\neq0$). 

There is perhaps no simpler context in which the superposition mystery
reveals itself than single photon interference, realized in the double-slit
experiment. In fact, in his Lectures in Physics, Richard Feynman famously
claimed that this experiment contains the \emph{only} mystery of quantum
mechanics \citep{feynman_feynman_2011}. Although there are other
mysteries, as \citet{silverman_more_1995} pointed out, the double-slit
experiment provides us with an understanding of some key aspects of
quantum mechanics. It may be the true mystery of quantum mechanics
lies in the idea that a property is not the same in different contexts
\citep{dzhafarov_all-possible-couplings_2013,markiewicz_unified_2013,howard_contextuality_2014,de_barros_unifying_2014},
and that such contexts (perhaps freely) chosen by an observer can
be spacelike separated \citep{bell_einstein-podolsky-rosen_1964,bell_problem_1966,greenberger_going_1989,de_barros_inequalities_2000,de_barros_probabilistic_2001,dzhafarov_probabilistic_2014}. 

One of the main ``disturbing'' mysteries of the double-slit system,
as described by Feynman, is the non-monotonic character of the probabilities
of detection. This was exactly what motivated Feynman to use negative
probabilities to describe quantum systems \citep{feynman_negative_1987}.
However, as Feynman remarked, such an approach did not seem to provide
any new insights into quantum mechanics. 

It is our goal here to show that we can indeed gain some insight by
using negative probabilities. This paper is organized the following
way. First, we introduce in Section \ref{sec:The-Mystery} a simplified
version of the double-slit experiment in the form of the Mach-Zehnder
interferometer. Then in Section \ref{sec:Negative-Probabilities},
we present a theory of negative probabilities. We then show, in Section
\ref{sec:The-Mach-Zehnder-Interferometer},  that proper probability
measures do not exist for the simultaneous measurements of the particle-
and wave-like properties, but negative probabilities do under certain
counterfactual conditions that are often assumed in experimental analyses.

\section{The Mystery of the double-slit Experiment\label{sec:The-Mystery}}

In this section we reproduce Feynman's discussion of the double-slit
experiment, and why he considered it mysterious \citep{feynman_feynman_2011}.
Feynman's argument involves the idea that classically we think of
systems in terms of two distinct and incompatible concepts, particles
or waves%
\footnote{You could also have fields, but in the context of the double-slit
experiment, as it will become clear later, the relevant property of
a field would be its spatial oscillations as a wave.%
}. Such concepts are incompatible because particles are localized and
waves are not. To see this, let us start with a point particle. In
classical mechanics, the main characteristic of particles is that
they are objects localized in space, and therefore can only interact
with other systems that are present in their localized position. For
example, let us consider a particle $P$ whose position at time $t$
is described by the position vector $\mathbf{r}_{P}\left(t\right)$.
At time $t_{0}$, $P$ can only interact with another physical entity
that is also at $\mathbf{r}_{P}\left(t_{0}\right)$, either another
particle $S$ such that $\mathbf{r}_{S}\left(t_{0}\right)=\mathbf{r}_{P}\left(t_{0}\right)$
or a field that is nonzero at $\mathbf{r}_{P}\left(t_{0}\right)$.
For instance, when a particle is subject to no external fields (of
course an idealization), such as gravity, it travels in a straight
line at constant speed, since no interaction is present. If this particle
then collides with another particle, say a constituent of a wall placed
in the way of the original particle, an interaction will appear%
\footnote{There is, of course, the obvious issue of how could this interaction
be relevant, given that it would occur with probability zero. %
}. But, as soon as the particle looses contact with the wall, the interaction
ceases to exist. In other words, particles interact locally. 

The second basic concept is that of waves. Historically, the physics
describing a point particle was extended to include the description
of continuous media, and, more important to our current discussion,
the vibrations of such media in the form of waves. Waves, therefore,
were considered vibrations of a medium made out of several point particles,
and the local interactions between two neighboring particles would
allow for a perturbation in one point of the medium to be propagated
to another point of the medium. Without going into the discussion
of the particulars of electromagnetic waves, the main point is that
because a single particle has an infinite number of neighbors, a disturbance
on its position propagates to \emph{all} of its neighbors and to \emph{all
}directions. Thus, the effect of such perturbation on particle $S'$
belonging to this medium due to a perturbation on particle $P'$ does
not depend on the direct contact of $S'$ and $P'$. More importantly,
such effect depends not only on $P'$, but also possibly on all other
particles that make up the medium, and also on all interactions or
boundary conditions that such particles need to satisfy. In other
words, waves interact non-locally. 

Going back to the double-slit experiment, let us analyze it from the
point of view of what one should expect to happen were it being modeled
with particles or with waves. Figure \ref{fig:Schematics-of-two-slit}
\begin{figure}
\begin{centering}
\includegraphics[scale=0.7]{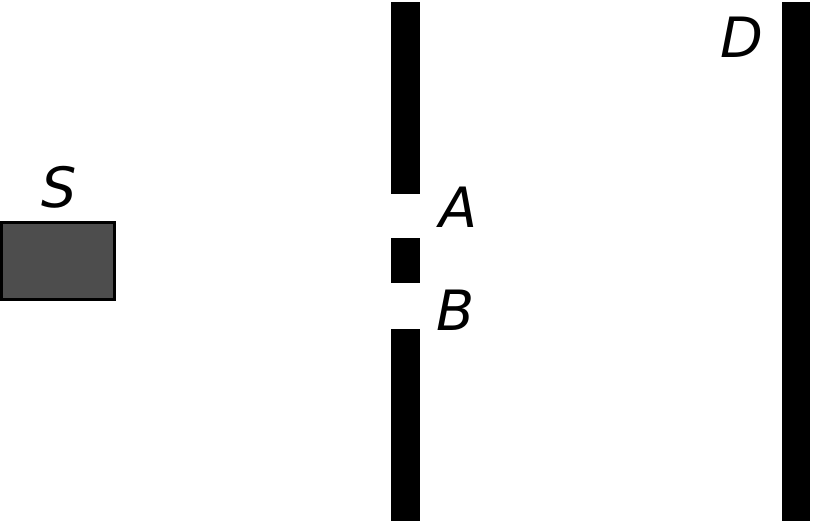}
\par\end{centering}

\protect\caption{\label{fig:Schematics-of-two-slit}Schematics of the double-slit experiment.
A source $S$ emits a physical object towards a barrier, where two
slits, $A$ and $B$, are cut to allow for its passage. Then, at a
screen $D$, the object is detected.}
\end{figure}
 shows a typical double-slit setup. We start with particles. Let us
assume that $S$ sends particles in random directions. A particle
leaving $S$ would interact only locally with the barrier and nothing
else. This means that in between $S$ and the barrier, this particle
travels along a straight line. Once it reaches the barrier, it either
goes through one of the slits and reaches $D$, or is reflected back
(or absorbed, depending on the barrier's properties). While going
through a slit, the particle may interact with the walls, perhaps
bouncing off of it, and therefore causing some scattering from the
direct path between $S$ and $D$. Thus, if we run this experiment
many times, we should expect the observed probability distribution
of particles on $D$ to be somewhat like as depicted in Figure \ref{fig:Probability-particle}.
\begin{figure}
\begin{centering}
\includegraphics[scale=0.2]{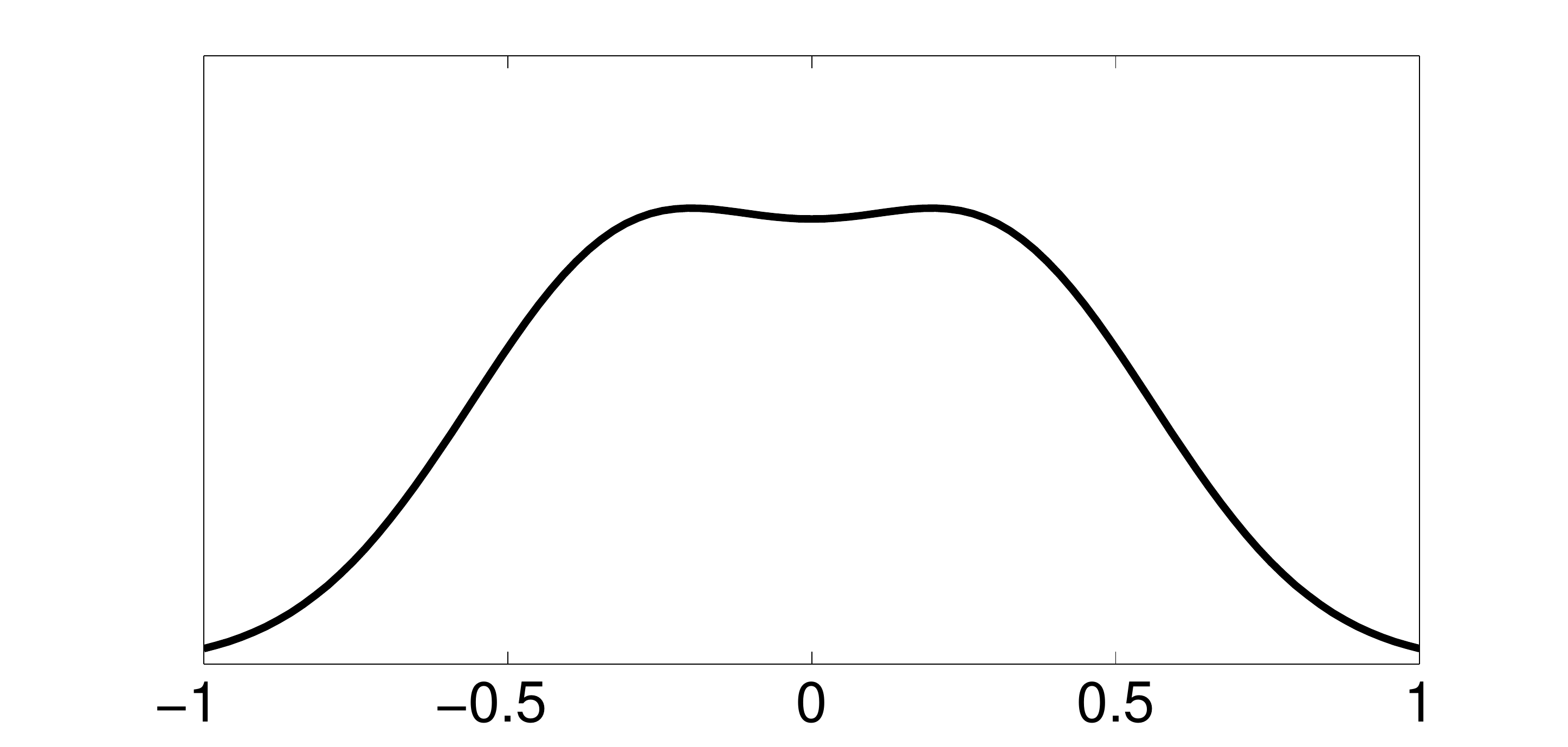}
\par\end{centering}

\protect\caption{\label{fig:Probability-particle}Probability density of observation
for the double-slit experiment, assuming a particle model. }
\end{figure}
 The resulting probability is simply the (normalized) sum of the probability
of a particle going through slit $A$ and $B$. 

A wave analysis of the experiment shows something quite different.
First, a wave is the result of a perturbation of a medium. In this
case, the source $S$ disturbs the medium, and such disturbance is
propagated in all directions. One characteristic of such propagation
is that its speed is dependent on the medium, and for the double slit
experiment this is reflected in the arrival of a wave crest (or valley)
in $A$ at the same time that a crest (or valley) arrives in $B$.
If $A$ and $B$ are small compared to the wavelengths, we can think
of them as secondary wave sources oscillating in phase. Thus, when
they arrive at $D$, in some places they will be in phase, whereas
in other places they will be out of phase. The result is the constructive
and destructive interference pattern that we know for waves, shown
in Figure \ref{fig:Probability-wave}. 
\begin{figure}
\begin{centering}
\includegraphics[scale=0.2]{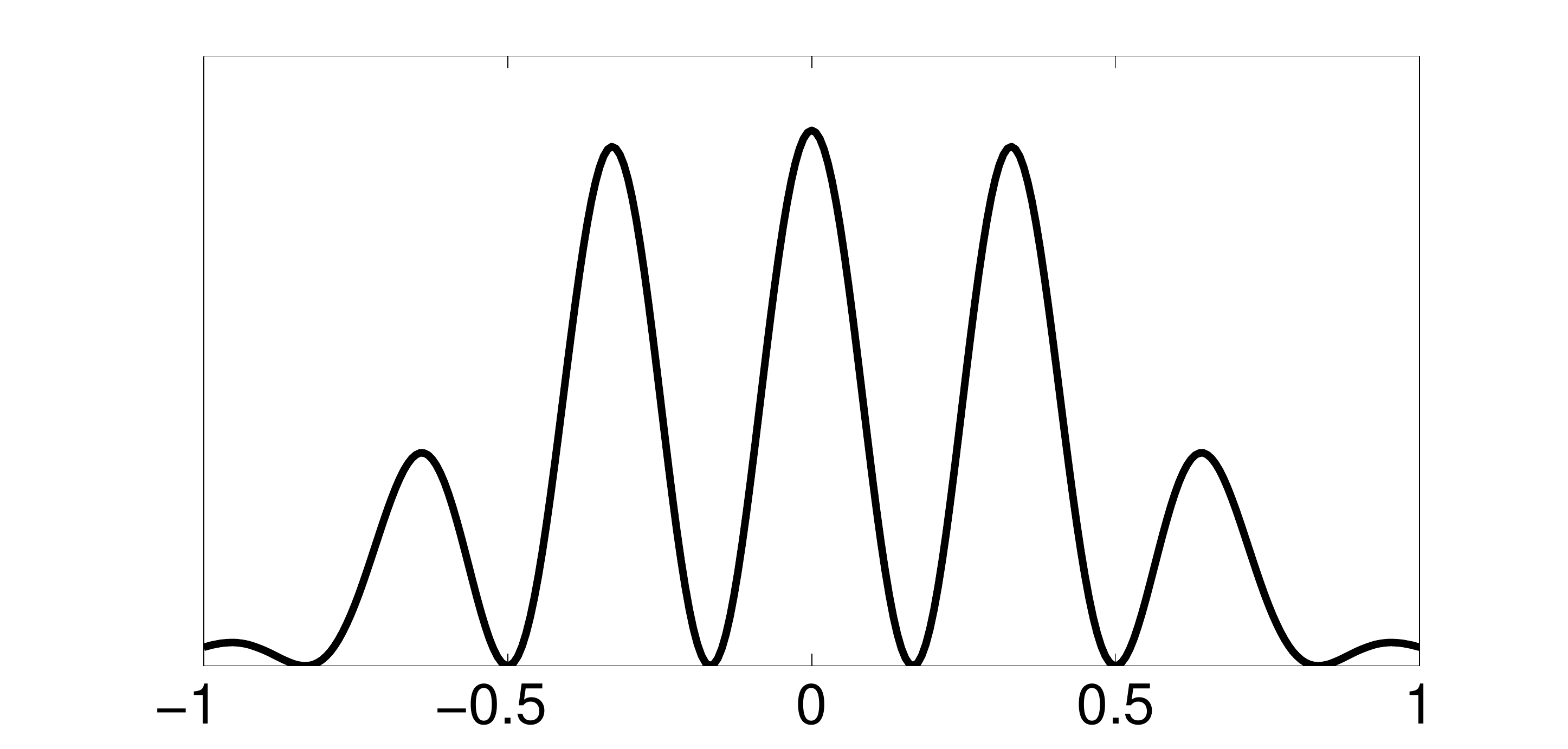}
\par\end{centering}

\protect\caption{\label{fig:Probability-wave}Intensity (arbitrary units) of the wave
at $D$ for the double-slit experiment. }
\end{figure}

Now for the puzzling aspect of it all: quantum systems, e.g. electrons,
have particle-like characteristics while at the same time being prone
to wave-like interference. Here is how those aspects manifest themselves.
Experimentally, electrons are particles. We know this because they
are localized, in the sense that when we measure an electron, it shows
up as a point in a fluorescent screen or a localized detector. This
is to be contrasted with waves, which are spread out, and therefore
have measurable components (i.e., momentum and energy) in more than
one place. Thus, as particles, when an electron leaves the source
$S$, it will either go to slit $A$ or $B$. Since its interactions
are solely local, if it goes through, say, $A$, it can only interact
with $A$, and not with $B$. Therefore, to a particle, it makes no
difference at all if slit $B$ is open or not when it goes through
$A$, and \emph{vice versa}. This is why we should expect a distribution
like that on Figure \ref{fig:Probability-particle}. The disturbing
fact is that an electron, if both slits are open, satisfies, after
many runs, the distribution given in Figure \ref{fig:Probability-wave}.
So, here is the puzzle. For a wave, the intensity is zero at several
points, e.g. at position $0.5.$ How can this zero intensity be understood?
How can the particle ``know'' (in the sense of interacting) about
$B$ if only interacting locally with $A$? To make this point even
clearer, we will examine this question in detail in Section \ref{sec:The-Mach-Zehnder-Interferometer},
using a simplified case of the double-slit experiment, the Mach-Zehnder
interferometer, together with a framework of extended probabilities.
But first, let us examine the concept of negative probabilities.

\section{Negative Probabilities\label{sec:Negative-Probabilities}}

In this Section we lay out the main relevant results and definitions
for negative probabilities. We start by first defining it in a way
that is related to Kolmogorov's \citeyearpar{kolmogorov_foundations_1950}
probability measure. We then prove some simple but relevant results. 
\begin{defn}
\label{def:Kolmogorov}Let $\Omega$ be a finite set, $\mathcal{F}$
an algebra over $\Omega$, and $p$ a real-valued functions, $p:\mathcal{F}\rightarrow\mathbb{R}$.
Then $\left(\Omega,\mathcal{F},p\right)$ is a probability space,
and $p$ a probability measure, if and only if:
\begin{eqnarray*}
\mbox{K1.} &  & 0\leq p\left(\left\{ \omega_{i}\right\} \right)\leq1,~~\forall\omega_{i}\in\Omega\\
\mbox{K2.} &  & p\left(\Omega\right)=1,\\
\mbox{K3.} &  & p\left(\left\{ \omega_{i},\omega_{j}\right\} \right)=p\left(\left\{ \omega_{i}\right\} \right)+p\left(\left\{ \omega_{j}\right\} \right),~~i\neq j.
\end{eqnarray*}
The elements $\omega_{i}$ of $\Omega$ are called \emph{elementary
probability events} or simply \emph{elementary events}. 
\end{defn}
Definition \ref{def:Kolmogorov} is the finite version of Kolmogorov's
standard definition of probability measure. In the usual definition,
$\Omega$ can be an infinite set, and then $\mathcal{F}$ needs to
be a $\sigma$-algebra. However, for the purpose of this article,
we will restrict our discussions to finite sets $\Omega$. For that
reason, we will also refer to $p$ as a \emph{proper probability distribution
or joint probability distribution}. 

It is a well know fact that for some systems it is not possible to
define a proper probability distribution. This is, in fact, the heart
of Bell's inequalities showing that for an EPR-Bohm type experiment
no local hidden-variable theories exist that are compatible with quantum
mechanics \citep{bell_einstein-podolsky-rosen_1964,bell_problem_1966}.
In fact, a hidden variable exists if and only if a joint probability
distribution exists \citep{suppes_when_1981,fine_hidden_1982}, and
Bell's inequalities are a necessary and sufficient condition for the
existence of a joint probability distribution \citep{fine_hidden_1982,suppes_collection_1996}.

To overcome this difficulty, the use of upper probabilities, where
K3 is modified to include subadditivity, has been proposed \citep{suppes_existence_1991,de_barros_probabilistic_2001,hartmann_entanglement_2010}.
The main reason for such a proposal is that quantum mechanics, as
suggested by Feynman's remarks, is nonmonotonic, and upper probabilities
offer a framework where nonmonotonicity can be described mathematically.
Thus, before we define negative probabilities, it is useful to start
with the more well-known theory of upper probabilities. 
\begin{defn}
\label{def:upper-probabilities}Let $\Omega$ be a finite set, $\mathcal{F}$
an algebra over $\Omega$, and $p^{*}$ a real-valued functions, $p^{*}:\mathcal{F}\rightarrow\mathbb{R}$.
Then $\left(\Omega,\mathcal{F},p^{*}\right)$ is an upper-probability
space, and $p^{*}$ an upper-probability measure, if and only if:
\begin{eqnarray*}
\mbox{U1.} &  & 0\leq p^{*}\left(\left\{ \omega_{i}\right\} \right)\leq1,~~\forall\omega_{i}\in\Omega\\
\mbox{U2.} &  & p^{*}\left(\Omega\right)=1,\\
\mbox{U3.} &  & p^{*}\left(\left\{ \omega_{i},\omega_{j}\right\} \right)\leq p^{*}\left(\left\{ \omega_{i}\right\} \right)+p^{*}\left(\left\{ \omega_{j}\right\} \right),~~i\neq j.
\end{eqnarray*}
\end{defn}
\begin{rem}
The main difference between upper and proper probabilities is the
substitution K3 for U3. 
\end{rem}

\begin{rem}
In many systems of interest, where the probabilities of elementary
events are computed from a set of given marginal probabilities, the
inequalities from U3 imply an underdetermination for the possible
values of the joint upper probability distribution for all $\omega_{i}\in\Omega$. 
\end{rem}

\begin{rem}
If follows from K2 and K3 that 
\[
\sum_{\omega_{i}\in\Omega}p\left(\left\{ \omega_{i}\right\} \right)=1,
\]
but because of U3 is 
\[
\sum_{\omega_{i}\in\Omega}p^{*}\left(\left\{ \omega_{i}\right\} \right)\geq1.
\]
 
\end{rem}
One of the main difficulties with upper probabilities is that, because
it uses subadditivity, it is very hard in practice to compute it.
Subaditivity also implies that a large number of different upper measures
exist, even when all moments are given. An usual approach is to request
$p^{*}$ to be as close as possible to a proper measure by minimizing
the value of $\sum_{\omega_{i}\in\Omega}p^{*}\left(\left\{ \omega_{i}\right\} \right)$,
which can be greater than one when no proper joint exists \citep{suppes_when_1981,suppes_existence_1991,fine_upper_1994,de_barros_probabilistic_2001,hartmann_entanglement_2010}. 

Another possible approach was proposed by \citet{feynman_negative_1987}
in connection to the two slit experiment: negative probabilities.
Though Feynman could not find any use for negative probability, recent
research has shown that there may be some advantage for using them
\citep{abramsky_sheaf-theoretic_2011,al-safi_simulating_2013,zhu_negative_2013,de_barros_decision_2014,oas_exploring_2014,abramsky_operational_2014,de_barros_negative_2014,de_barros_unifying_2014}.
To define negative probabilities, we first need to set forth a description
of certain systems where no proper probability distribution exists.
This is the goal of the following definition. 
\begin{defn}
Let $\Omega$ be a finite set, $\mathcal{F}$ an algebra over $\Omega$,
and let $\left(\Omega_{i},\mathcal{F}_{i},p_{i}\right)$, $i=1,\ldots,n$,
a set of $n$ probability spaces, $\mathcal{F}_{i}\subseteq\mathcal{F}$
and $\Omega_{i}\subseteq\Omega$. Then $\left(\Omega,\mathcal{F},p\right)$,where
$p$ is a real-valued function, $p:\mathcal{F}\rightarrow\left[0,1\right]$,
$p\left(\Omega\right)=1$, is \emph{compatible} with the probabilities
$p_{i}$'s if and only if 
\[
\forall\left(x\in\mathcal{F}_{i}\right)\left(p_{i}\left(x\right)=p\left(x\right)\right).
\]
Furthermore, the marginals $p_{i}$ are \emph{viable} if and only
if $p$ is a probability measure. \end{defn}
\begin{rem}
Intuitively, we can think of the $p_{i}$'s as observable marginals.
The definition above says that such marginals are \emph{viable} if
it is possible to sew them together to produce a larger probability
function over the whole space $\Omega$ (in the same spirit of \citet{dzhafarov_all-possible-couplings_2013,dzhafarov_contextuality_2014,de_barros_unifying_2014}).\emph{
}Our definition is an extension of \citet{halliwell_negative_2013},
as we consider not only the case where viable distributions exist,
but also when they do not. 
\end{rem}
For some experimental situations, such as the EPR-Bell setup, the
marginals are not viable, but are compatible with a $p$ that has
the characteristic of being negative for some elements of $\Omega$
(but not negative for the observable marginals). This motivates the
following definition of a $p$ that may take negative values. 
\begin{defn}
\label{def:negative-probabilities}Let $\Omega$ be a finite set,
$\mathcal{F}$ an algebra over $\Omega$, $P$ and $P'$ real-valued
functions, $P:\mathcal{F}\rightarrow\mathbb{R}$, $P':\mathcal{F}\rightarrow\mathbb{R}$,
and let $\left(\Omega_{i},\mathcal{F}_{i},p_{i}\right)$, $i=1,\ldots,n$,
a set of $n$ probability spaces, $\mathcal{F}_{i}\subset\mathcal{F}$
and $\Omega_{i}\subseteq\Omega$. Then $\left(\Omega,\mathcal{F},P\right)$
is a negative probability space, and $P$ a negative probability,
if and only if $\left(\Omega,\mathcal{F},P\right)$ is compatible
with the probabilities $p_{i}$'s and
\begin{eqnarray*}
\mbox{N1.} &  & \forall\left(P'\right)\left(\sum_{\omega_{i}\in\Omega}\left|P\left(\left\{ \omega_{i}\right\} \right)\right|\leq\sum_{\omega_{i}\in\Omega}\left|P'\left(\left\{ \omega_{i}\right\} \right)\right|\right)\\
\mbox{N2.} &  & \sum_{\omega_{i}\in\Omega}P\left(\left\{ \omega_{i}\right\} \right)=1\\
\mbox{N3.} &  & P\left(\left\{ \omega_{i},\omega_{j}\right\} \right)=P\left(\left\{ \omega_{i}\right\} \right)+P\left(\left\{ \omega_{j}\right\} \right),~~i\neq j.
\end{eqnarray*}

\end{defn}
In the above definition, we replace axiom K1 of nonnegativity with
a minimization of the L1 norm of the function $P$. Intuitively, as
with uppers, we seek a quasi-probability distribution that is as close
to a proper distribution as possible. Furthermore, the departure from
such proper distributions, which would have no negative numbers, motivates
the following definition of $M^{*}$ as a measure of this departure.
Throughout this paper we use $p$ for proper probability measures
(Definition \ref{def:Kolmogorov}), $p^{*}$ for upper and lower probabilities
(Definition \ref{def:upper-probabilities}), and $P$ for negative
probabilities (Definition \ref{def:negative-probabilities}). 
\begin{defn}
Let $\left(\Omega,\mathcal{F},P\right)$ be a negative probability
space. Then, the \emph{minimum L1 probability norm}, denoted\emph{
$M^{*}$, }or simply \emph{minimum probability norm}, is given by
$M^{*}=\sum_{\omega_{i}\in\Omega}\left|P\left(\left\{ \omega_{i}\right\} \right)\right|$.\end{defn}
\begin{prop}
\label{prop:NegativeLeadsToKolmog}Let $\left(\Omega,\mathcal{F},P\right)$
be a negative probability space and $\left(\Omega,\mathcal{F},p\right)$
a (Kolmogorov) probability space. Then $p=P$ iff $M^{*}=\sum p_{i}$. \end{prop}
\begin{proof}
Let us start with $M^{*}=\sum p_{i}$. It follows from it that all
elementary events satisfy the condition $0\le P\left(\left\{ \omega_{i}\right\} \right)\le1$,
which is K1. Together with N2 and N3, then $P$ is also a probability
measure, and $\left(\Omega,\mathcal{F},P\right)$ a probability space.
Now, if $\left(\Omega,\mathcal{F},p\right)$ is a probability space,
it follows from K1 that $M^{*}=\sum_{\omega_{i}\in\Omega}\left|P\left(\left\{ \omega_{i}\right\} \right)\right|=\sum p_{i}$. \end{proof}
\begin{rem}
\label{Rem:NegAxiomImplyKolmo}Proposition \ref{prop:NegativeLeadsToKolmog}
tells us that axioms N1-N3 include, as a special case, K1-K3. In other
words, in the special case when a proper Kolmogorovian distribution
exists ($M^{*}=\sum p_{i}$), $P$ coincides with $p$. 
\end{rem}
We end this section with one last definition that is relevant to physical
systems.
\begin{defn}
Let $\Omega$ be a finite set, $\mathcal{F}$ an algebra over $\Omega$,
and let $\left(\Omega_{i},\mathcal{F}_{i},p_{i}\right)$, $i=1,\ldots,n$,
a collection of $n$ probability spaces, $\mathcal{F}_{i}\subseteq\mathcal{F}$
and $\Omega_{i}\subseteq\Omega$. Then the probabilities $p_{i}$
are \emph{contextually biased}%
\footnote{Here we adopt and adapt the terminology of \citet{dzhafarov_contextuality_2014}.%
} if there exists an $a$ in $\mathcal{F}_{i}$ and in $\mathcal{F}_{j}$,
$i\neq j$, $b\neq a\neq b'$, $\sum_{\forall b\in\mathcal{F}_{j}}p\left(a\cap b\right)\neq\sum_{\forall b'\in\mathcal{F}_{i}}p\left(a\cap b'\right)$. \end{defn}
\begin{rem}
In physics, for multipartite systems, this definition is equivalent
to the no-signaling condition. 

\end{rem}
\begin{prop}
Let $\Omega$ be a finite set, $\mathcal{F}$ an algebra over $\Omega$,
and let $\left(\Omega_{i},\mathcal{F}_{i},P_{i}\right)$, $i=1,\ldots,n$,
a set of $n$ probability spaces, $\mathcal{F}_{i}\subseteq\mathcal{F}$
and $\Omega_{i}\subseteq\Omega$. The probabilities $P_{i}$ are not
\emph{contextuality biased} if and only if there exists a negative
probability $\left(\Omega,\mathcal{F},P\right)$ compatible with the
$p_{i}$'s. \end{prop}
\begin{proof}
See \citet{al-safi_simulating_2013,oas_exploring_2014,abramsky_operational_2014}
for different proofs. 
\end{proof}

\subsection{Interpretations of Negative Probabilities}

As we mentioned above, both Dirac and Feynman saw negative probabilities
as computational devices. Though we can take such pragmatic view,
as negative probabilities help explore certain situations of interest
in quantum mechanics (see \citet{oas_exploring_2014} for an example),
the question still remains as to their meaning. Here we discuss some
proposals on how to interpret negative probabilities. 

Let us start with the interpretation of negative probabilities in
terms of two disjoint probability measures, $\mu^{+}$ and $\mu^{-}$,
initially suggested by \citet{burgin_interpretations_2010,burgin_negative_2012}
and then expanded and formalized by \citet{abramsky_operational_2014}.
Here we follow the interpretation as presented by \citet{abramsky_operational_2014}.
The main idea of this interpretation comes from the well-known fact
(see \citet{rao_theory_1983}) that it is possible to decompose a
signed measure $\mu$ into two non-negative ones, $\mu^{-}$ and $\mu^{+}$,
such that 
\[
\mu=\mu^{+}-\mu^{-}.
\]
Following this idea, \citet{abramsky_operational_2014} creates two
copies of the sample space, namely $\Omega\times\left\{ +,-\right\} $,
giving them a new dimension corresponding to $+$ or $-$. For example,
for the case of three random variables $\mathbf{X}$, $\mathbf{Y}$,
and $\mathbf{Z}$, the set of all elementary events would be $\left\{ \omega_{xyz},\omega_{xy\overline{z}},\ldots,\omega_{\overline{x}\overline{y}z},\omega_{\overline{x}\overline{y}\overline{z}}\right\} $,
whereas the expanded set $\Omega\times\left\{ +,-\right\} $ would
have as elementary events $\left\{ \omega_{xyz+},\ldots,\omega_{\overline{x}\overline{y}\overline{z}+},\omega_{xyz-},\ldots,\omega_{\overline{x}\overline{y}\overline{z}-}\right\} $.
Because of the above decomposition, they define a probability over
the set $+$ and $-$ such that 
\[
P=p^{+}-p^{-},
\]
where now $P$ is the negative probability and $p^{+}$ and $p^{-}$
can be interpreted as proper probability distributions over $\left(\Omega,\pm\right)$.
To interpret $P$, \citet{abramsky_operational_2014} proposes an
effect akin to interference. When we observe an event, say corresponding
to the element $\omega_{xyz+}$, we use $p^{+}$ as the distribution
to create our data table, and similarly for $\omega_{xyz-}$. However,
the counting due to a $-$ element can annihilate a counting for a
$+$ element, and vice versa. In a certain sense, this interpretation
of negative probabilities is conceptually similar to some hidden variable
approaches in the literature, as for example the virtual photon model
of Suppes and de Barros \citep{suppes_violation_1996,suppes_random-walk_1994,suppes_proposed_1996,suppes_particle_1996,suppes_diffraction_1994,suppes_collection_1996}
or the ESR model \citep{garola_esr_2009,sozzo_hilbert_2010,garola_generalized_2011,garola_modified_2011,garola_finite_2014},
to cite a few. In these approaches, an underlying hidden process can
erase an outcome that would be possible if it were not for the interference
of non-observable events. However, the problem with this interpretation
is that, even though it is based on a frequentist view, it does not
provide a way of counting actual observable clicks on a measurement
device and interpret them as negative probabilities; in other words,
it assumes some non-accessible reality. 

We now turn to another frequentist interpretation of negative probabilities,
this one proposed by Khrennikov \citep{khrennikov_p-adic_1993,khrennikov_p-adic_1993-1,khrennikov_p-adic_1994,khrennikov_discrete_1994,khrennikov_interpretations_2009}.
Khrennikov starts with the idea that, in the frequentist view, the
probability of an event is defined as by the number of times such
an event occurs in an infinite sequence of possible outcomes or ensembles.
Following \citet{khrennikov_interpretations_2009}, let $S_{N}$ be
a sequence of $N$ ensembles with, $S_{N}=\left\{ s_{1},s_{2},\ldots,s_{N}\right\} $.
For each of the ensembles $s_{i}$, one can ask whether the property
represented by the random variable $\mathbf{A}$ has the value $a$
or not, and let $S\left(\mathbf{A}=a\right)$ be the subset of all
ensembles such that $\mathbf{A}=a$. Then, in the standard frequentist
interpretation, the probability $p\left(\mathbf{A}=a\right)$ is given
by
\begin{equation}
p\left(\mathbf{A}=a\right)=\lim_{N\rightarrow\infty}\frac{\left|S\left(\mathbf{A}=a\right)\cap S_{N}\right|}{\left|S\right|},\label{eq:frequentist-probability}
\end{equation}
where $\left|\cdot\right|$ represents the cardinality of a set. Khrennikov
then argues that there are ensembles for which the limit in (\ref{eq:frequentist-probability})
does not converge, and for such cases negative probabilities can be
obtained as the result of a regularization procedure or order. In
such sense, negative probabilities come as the result of quasi-random
sequences that violate the principle of statistical stabilization
\citep{khrennikov_p-adic_1993}. Khrennikov then proposes to generalize
probabilities coming from sequences that violate the principle of
statistical stabilization as measures taking values not only on the
field $\mathbb{R}$, but also on the $p$-adic extensions of the set
of rationals $\mathbb{Q}$, i.e. $\mathbb{Q}_{p}$. We recall that
$\mathbb{R}$ is defined, through Cauchy sequences, as the completion
of $\mathbb{Q}$ under the Euclidian norm. Similarly, $\mathbb{Q}_{p}$
is the completion of $\mathbb{Q}$ under the $p$-adic norm (see \citet{khrennikov_interpretations_2009}
for a clear exposition of $\mathbb{Q}_{p}$ and its properties). Once
he does that, he shows that certain sequences that have probability
zero in the sense of (\ref{eq:frequentist-probability}) would have
negative probabilities in their $p$-adic extension, whereas sequences
of probability one would have $p$-adic values greater than one. Thus,
according to Khrennikov, we can interpret negative (greater than one)
probabilities as events of probability zero (one) for sequences that
violate the principle of statistical stabilization. Here we note that
contextual random variables are quasi-random, and violate the principle
of statistical stabilization. 

We now turn to the meaning of the minimum L1 norm we propose for negative
probabilities. Similarly to negative probabilities, the sub or super
additivity of upper and lower probabilities allows for a large number
of solutions to the joint probability that is consistent with the
marginals. One possibility is to think of upper and lowers as subjective
measures of belief based on inconsistent information \citep{suppes_existence_1991}.
As such, it can be argued that, since upper and lowers do not add
to one as standard probabilities do, one should choose among the many
different distributions those whose sums are as close to one as possible.
This is, in a certain sense, similar to what the minimum L1 norm does
for negative probabilities. As such, this norm, which quantifies how
much a negative probability deviates from a proper probability, provides
us a measure of how inconsistent the correlations between random variables
are \citep{de_barros_decision_2014}. 

We end this section with one last general comment. Instead of using
negative probabilities, it is possible to simply extend the probability
space such that when we talk about correlations between experimentally
observable variables, as proposed by \citet{dzhafarov_selectivity_2012,dzhafarov_all-possible-couplings_2013}.
To understand this point, imagine we start with three variables $\mathbf{X}$,
\textbf{$\mathbf{Y}$}, and $\mathbf{Z}$, as in the above example.
Instead of thinking of them as three variables, we could think of
them as six, one for each experimental context: $\mathbf{X}_{Y}$,
$\mathbf{X}_{Z}$, ..., $\mathbf{Z}_{Y}$. It is easy to show that
in some important physical examples, such as the famous Bell-EPR setup,
such extension of the probability space is sufficient to grant the
existence of a joint probability distribution, but at the cost of
having $\mathbf{X}_{Y}\neq\mathbf{X}_{Z}$. Thus, the apparent inconsistencies
mentioned in the previous paragraph could be argued to come from an
identity assumption for the random variables: that a random variable
remains the same in different contexts \citep{dzhafarov_random_2013,dzhafarov_all-possible-couplings_2013,dzhafarov_contextuality_2014,dzhafarov_qualified_2014}.
As such, the minimum norm could be interpreted as a measure of contextuality
\citep{de_barros_unifying_2014,de_barros_negative_2014,oas_exploring_2014}.

\section{The Mach-Zehnder Interferometer and Negative Probabilities\label{sec:The-Mach-Zehnder-Interferometer}}

Now that we saw the basic relationships between negative probabilities
and upper and Kolmogorovian probabilities, we turn our attention to
the two slit experiment its simplified version of the Mach-Zehnder
interferometer, schematically shown in Figure \ref{fig:Mach-Zehnder-interferometer}.
\begin{figure}
\begin{centering}
\includegraphics{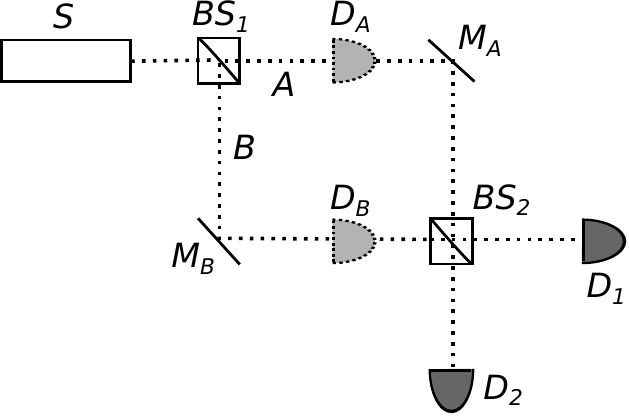}
\par\end{centering}

\protect\caption{\label{fig:Mach-Zehnder-interferometer}Schematics of a Mach-Zehnder
interferometer. A light beam from a source $S$ is divided into two
equal-intensity beams by the beamsplitter $BS_{1}$. The beams are
reflected by mirrors $M_{A}$ and $M_{B}$, and then recombined by
the beamsplitter $BS_{2}$. Photons from $S$ are then detected by
photodetectors $D_{1}$ or $D_{2}$. The count rates on $D_{1}$ and
$D_{2}$ depend on the geometry of the system, in particular the optical
distances between $BS_{1}$, $BS_{2}$, and $M_{A}$ and $M_{B}$.
In the which-path version of this experiment, detectors $D_{A}$ and/or
$D_{B}$ may be placed on each arm of the interferometer to determine
the trajectory of the photon. }
\end{figure}
 We will not attempt to give a full quantum-mechanical description
of this experiment, but instead focus on an elementary representation
of the main ideas behind it. Intuitively, each arm of the interferometer,
$A$ and $B$, corresponds to a possible path the particle can take,
which is equivalent to a particle going through one slit or the other
on the two slit experiment. However, the Mach-Zehnder differs from
the double-slit experiment as the particle only has two possible outcomes,
a detection in $D_{1}$ or $D_{2}$, as opposed to an infinite number
of points on a screen. Such two possible outcomes could be thought
as two points on the screen on the screen corresponding to a maximum
and minimum intensities in the interference pattern. 

To elaborate on the analogy with the two slit experiment, let us think
about the experiment in terms of particles. First, we can select an
interferometer setup such that we have constructive interference in
$D_{1}$ and destructive in $D_{2}$. In other words, the lengths
of the interferometer arms are chosen such that $P(\mathbf{D}_{1}=1)=1$
and $P(\mathbf{D}_{2}=1)=0,$ where $\mathbf{D}_{1}$ ($\mathbf{D}_{2}$)
is a $\pm1$-valued random variable representing a detection on $D_{1}$
($D_{2}$) when its value is $1$ and no detection when $-1$. From
now on we will use the standard notation $p_{d_{1}}=P(\mathbf{D}_{1}=1)$,
$p_{\overline{d}_{1}}=P(\mathbf{D}_{1}=-1)$, $p_{\overline{d}_{2}}=P(\mathbf{D}_{2}=-1)$
and so on. With this notation, our interferometer is such that $p_{d_{1}}=p_{\overline{d}_{2}}=1$
and $p_{\overline{d}_{1}}=p_{d_{2}}=0$. 

Now that the interferometer is set up, let us examine the two possible
classical models (according to Feynman) behind it: the wave and particle
models. We start with the wave point of view. Let $\psi=A\cos\left(\omega t\right)$
represent a coherent wave arriving at the beam splitter $BS_{1}$
at time $t$ and being split in both directions, $A$ \emph{and }$B$.
The wave going through $A$ is unchanged by $BS_{1}$, and arrives
at $M_{A}$ as $\frac{A}{2}\cos\left(\omega t+\phi_{1}\right)$, where
$\phi_{1}$ is a phase that depends on the geometry of the interferometer,
specifically on the distance between $BS_{1}$ and $M_{A}$. At $M_{A}$
it becomes $-\frac{A}{2}\sin\left(\omega t+\phi_{1}\right)$ due to
a $\pi/2$ phase shift upon reflection, and arrives at $BS_{2}$ as
$-\frac{A}{2}\sin\left(\omega t+\phi_{1}+\phi_{2}\right)$. For the
wave going through $B$, it arrives at $M_{B}$ as $-\frac{A}{2}\sin\left(\omega t+\phi_{2}\right)$
and at $BS_{2}$ as $-\frac{A}{2}\cos\left(\omega t+\phi_{2}+\phi_{1}\right)$,
where we assume for the geometry that distance between $BS_{1}$ and
$M_{B}$ is the same as the distance between $M_{A}$ and $BS_{2}$
(and similarly for $BS_{1}$to $M_{A}$ and $M_{B}$ to $BS_{2}$).
The beam splitter $BS_{2}$ now recombines the two waves coming from
$A$ and \textbf{$B$}, and the outputs on $D_{1}$ and $D_{2}$ are
the superposition of those waves. In other words, 
\begin{eqnarray}
\psi_{D_{1}} & = & -\frac{A}{2}\sin\left(\omega t+\phi_{1}+\phi_{2}+\frac{\pi}{2}\right)-\frac{A}{2}\cos\left(\omega t+\phi_{2}+\phi_{1}\right)\label{eq:psid1}\\
 & = & -\frac{A}{2}\cos\left(\omega t+\phi\right)-\frac{A}{2}\cos\left(\omega t+\phi\right)\nonumber \\
 & = & -A\cos\left(\omega t+\phi\right),\nonumber 
\end{eqnarray}
where the first term on the rhs is the reflected wave from $A$, and
$\phi=\phi_{1}+\phi_{2}$. For $D_{2}$ we obtain, with now the wave
$B$ getting a phase of $\pi/2$,
\begin{eqnarray}
\psi_{D_{2}} & = & -\frac{A}{2}\sin\left(\omega t+\phi\right)-\frac{A}{2}\cos\left(\omega t+\phi+\frac{\pi}{2}\right)\label{eq:psid2}\\
 & = & -\frac{A}{2}\sin\left(\omega t+\phi\right)+\frac{A}{2}\sin\left(\omega t+\phi\right)=0.
\end{eqnarray}
We can now compute the mean intensity of the entering wave, $\psi$,
\[
I_{S}=\langle\psi^{2}\rangle_{t}=\frac{A^{2}}{2},
\]
where 
\[
\langle f\rangle_{t}=\frac{1}{T}\lim_{\omega T\gg1}\int_{t}^{t+T}f\, dt'
\]
 represents the time average. The intensity at $D_{1}$, is
\begin{eqnarray*}
I_{D_{1}} & = & \frac{A^{2}}{2},
\end{eqnarray*}
whereas the intensity at $D_{2}$ is 
\begin{eqnarray*}
I_{D_{2}} & = & 0,
\end{eqnarray*}
consistent with the value for the source $S$. These particular values
for the intensities at $D_{1}$ and $D_{2}$ present the highest contrast
between the intensities at each detector, or, as it is often referred,
maximum visibility (of interference). So, to summarize, according
to the wave model we see no wave energy arriving at $D_{2}$ because
of destructive interference due to the relative phases of the different
paths the wave traveled. 

Now let us examine the view that photons are particles, and let us
assume that we can control the intensity of the source such that one
particle at a time goes through the interferometer. A particle comes
out of the source $S$ and enters the interferometer through the beam
splitter $BS_{1}$. Beam splitters divide beams into two equal intensity
ones. This translates into a particle having probability $1/2$ of
going to either arm $A$ or $B$. For the sake of argument, let us
assume that the particle went into arm $A$. Once it leaves the interferometer,
it is reflected at mirror $M_{A}$ and reaches another beam splitter
$BS_{2}$. Once again, it has probability $1/2$ of going on either
direction, since it interacts only locally with $BS_{2}$. In other
words, it cannot possibly have any information about the geometry
of path $B$, or even if it is not simply closed with the presence
of a physical barrier. Therefore, the probability of this particle
reaching $D_{1}$ is the same as $D_{2}$, and it equals $1/2$. The
same analysis can be applied to the photon going through arm $B$.
Therefore, from a particle point of view, $p_{d_{1}}=p_{d_{2}}=1/2$.
This is in stark contradiction with the wave result. 

In the standard interpretation of quantum mechanics, this contradiction
is resolved by stating that one cannot simultaneously assign two complimentary
properties to a quantum system. In the above case, we cannot assign
the property of going through path $A$ or $B$ (which is what happens
if we have a particle). To be able to say that a particle went through
$A$ or $B$, we need to actually place a detector $D_{A}$ and $D_{B}$
in the paths. At the same time, if we place such detectors, we destroy
the wave-like behavior, and its associated probabilities at $D_{1}$
and $D_{2}$. If the detectors simply destroy the particle, then we
have obviously an impossibility in obtaining the joint probability
distribution in a trivial way, as we can show by the following simple
example (which we also spell out in more detail below). 

Let $D_{X}$ and $D_{Y}$ be two detectors that absorb photons, and
let us put $D_{Y}$ at the end of a source $S$ that produces photons.
So, if a photon is emitted, the probability of observing it is $p_{d_{y}}=1$.
However, if we put $D_{X}$ in between $S$ and $D_{Y}$, we will
have that $p_{d_{x}}=1$ and $p_{d_{y}}=0$. The observable terms
are the following.
\[
p_{d_{x}d_{y}}=p_{\overline{d}_{x}\overline{d}_{y}}=p_{\overline{d}_{x}d_{y}}=0,
\]
\[
p_{d_{x}\overline{d}_{y}}=1.
\]
But this leads to a contradiction, as 
\[
p_{d_{y}}=1=p_{d_{x}d_{y}}+p_{\overline{d}_{x}d_{y}}=0.
\]
That contradiction comes from an obvious reason: we have different
experiments, and therefore the random variable $\mathbf{D}_{Y}$ representing
a measurement in one experiment cannot be the same as the $\mathbf{D}_{Y}$
in the other experiment. The assumption of the existence of a joint
distribution is equivalent to the assumption that both $\mathbf{D}_{Y}$'s
are the same.

In the Mach-Zehnder, an analogous case to the example above would
be the following. With the setup in Figure \ref{fig:Mach-Zehnder-interferometer},
we split the experiment into two types: destructive and non-destructive
measurements. A destructive measurement happens when the observed
system is not available for any other measurements afterward. For
example, in many photodetection apparatuses, the photon is absorbed
by the device and ceases to exist. A non-destructive measurement is
the one where the system is available for later measurements. For
each type of experiment, there are four possible experimental conditions,
which we label as Case 1 to Case 8. We start with destructive measurements. 
\begin{description}
\item [{Case~1~($D_{1}$,~$D_{2}$~only)}] This case corresponds to
the standard Mach-Zehnder with no \emph{which-path} information, since
no detector is put on either arm of the interferometer. Thus, a joint
probability distribution exists for all the random variables involved.
When this is the case, we have that $p_{d_{1}\overline{d}_{2}}=1$
and $p_{d_{1}d_{2}}=p_{\overline{d}_{1}d_{2}}=p_{\overline{d}_{1}\overline{d}_{2}}=0$.
Here $p_{d_{1}d_{2}}$ and $p_{\overline{d}_{1}\overline{d}_{2}}$
are set to zero almost by definition, as we are considering cases
where we have one and only one photo-detection. 
\item [{Case~2~($D_{1}$,~$D_{2}$,~$D_{A}$)}] In this case, if we
have a detection on $D_{A}$, we have no detection on $D_{1}$ and
$D_{2}$ (intuitively, the photon was absorbed by the detector). On
the other hand, if we have no detection on $D_{A}$, $D_{1}$ and
$D_{2}$ are equiprobable, since the interference effects are destroyed
by the presence of a detection. Thus, $p_{d_{a}d_{1}d_{2}}=p_{d_{a}d_{1}\overline{d}_{2}}=p_{d_{a}\overline{d}_{1}\overline{d}_{2}}=p_{d_{a}\overline{d}_{1}d_{2}}=$
$p_{\overline{d}_{a}d_{1}d_{2}}=p_{\overline{d}_{a}\overline{d}_{1}\overline{d}_{2}}=0$,
and $p_{\overline{d}_{a}d_{1}\overline{d}_{2}}=p_{\overline{d}_{a}\overline{d}_{1}d_{2}}=\frac{1}{2}$.
\item [{Case~3~($D_{1}$,~$D_{2}$,~$D_{B}$)}] Similarly to Case 2
above, here $p_{d_{b}d_{1}d_{2}}=p_{d_{b}d_{1}\overline{d}_{2}}=p_{d_{b}\overline{d}_{1}\overline{d}_{2}}=p_{d_{b}\overline{d}_{1}d_{2}}=$
$p_{\overline{d}_{b}d_{1}d_{2}}=p_{\overline{d}_{b}\overline{d}_{1}\overline{d}_{2}}=0$,
and $p_{\overline{d}_{b}d_{1}\overline{d}_{2}}=p_{\overline{d}_{b}\overline{d}_{1}d_{2}}=\frac{1}{2}$.
\item [{Case~4~($D_{1}$,~$D_{2}$,~$D_{A}$,~and~$D_{B}$)}] This
simply tells us that we can only observe in $A$ or $B$, but nowhere
else, since the detectors in $A$ and $B$ destroy the photon, not
allowing it to reach $D_{1}$ or $D_{2}$. Then $p_{d_{a}d_{b}d_{1}d_{2}}=p_{d_{a}d_{b}\overline{d}_{1}d_{2}}=p_{d_{a}d_{b}d_{1}\overline{d}_{2}}=p_{d_{a}d_{b}\overline{d}_{1}\overline{d}_{2}}=0$,
$p_{\overline{d}_{a}d_{b}d_{1}d_{2}}=p_{\overline{d}_{a}d_{b}\overline{d}_{1}d_{2}}=p_{\overline{d}_{a}d_{b}d_{1}\overline{d}_{2}}=0$,
$p_{d_{a}\overline{d}_{b}d_{1}d_{2}}=p_{d_{a}\overline{d}_{b}\overline{d}_{1}d_{2}}=p_{d_{a}\overline{d}_{b}d_{1}\overline{d}_{2}}=0$,
$p_{\overline{d}_{a}\overline{d}_{b}d_{1}d_{2}}=p_{\overline{d}_{a}\overline{d}_{b}\overline{d}_{1}d_{2}}=p_{\overline{d}_{a}\overline{d}_{b}d_{1}\overline{d}_{2}}=p_{\overline{d}_{a}\overline{d}_{b}\overline{d}_{1}\overline{d}_{2}}=0$,
and $p_{\overline{d}_{a}d_{b}\overline{d}_{1}\overline{d}_{2}}=p_{d_{a}\overline{d}_{b}\overline{d}_{1}\overline{d}_{2}}=\frac{1}{2}$. 
\end{description}
It is easy to see that we have inconsistencies between the random
variables for each case, because Case 4 gives us a joint probability
distribution for all observables that is inconsistent with Case 1.
This simply tells us that each experimental context gives different
distributions to the random variables, as the marginal expectations
are different for each experimental condition. 

For the non-destructive measurement we have the following experimental
outcomes. 
\begin{description}
\item [{Case~5~($D_{1}$,~$D_{2}$~only)}] This is clearly identical
to Case 1, where $p_{d_{1}\overline{d}_{2}}=1$ and $p_{d_{1}d_{2}}=p_{\overline{d}_{1}d_{2}}=p_{\overline{d}_{1}\overline{d}_{2}}=0$. 
\item [{Case~6~($D_{1}$,~$D_{2}$,~$D_{A}$)}] In this case, if we
have a detection on $D_{A}$, but there will also be a detection on
either $D_{1}$ or $D_{2}$. Furthermore, regardless of the outcomes
on $D_{A}$, detections on $D_{1}$ and $D_{2}$ are equiprobable,
since the interference effects are destroyed by the presence of a
detection. Thus, $p_{d_{a}d_{1}d_{2}}=p_{d_{a}\overline{d}_{1}\overline{d}_{2}}=$
$p_{\overline{d}_{a}d_{1}d_{2}}=p_{\overline{d}_{a}\overline{d}_{1}\overline{d}_{2}}=0$,
and $p_{d_{a}d_{1}\overline{d}_{2}}=p_{d_{a}\overline{d}_{1}d_{2}}=p_{\overline{d}_{a}d_{1}\overline{d}_{2}}=p_{\overline{d}_{a}\overline{d}_{1}d_{2}}=\frac{1}{4}$.
\item [{Case~7~($D_{1}$,~$D_{2}$,~$D_{B}$)}] Similarly to Case 2
above, here $p_{d_{b}d_{1}d_{2}}=p_{d_{b}\overline{d}_{1}\overline{d}_{2}}=$
$p_{\overline{d}_{b}d_{1}d_{2}}=p_{\overline{d}_{b}\overline{d}_{1}\overline{d}_{2}}=0$,
and $p_{d_{b}d_{1}\overline{d}_{2}}=p_{d_{b}\overline{d}_{1}d_{2}}=p_{\overline{d}_{b}d_{1}\overline{d}_{2}}=p_{\overline{d}_{b}\overline{d}_{1}d_{2}}=\frac{1}{4}$.
\item [{Case~8~($D_{1}$,~$D_{2}$,~$D_{A}$,~and~$D_{B}$)}] This
simply tells us that we can only observe in $A$ or $B$, but nowhere
else, since the detectors in $A$ and $B$ absorb the photon, not
letting it reach $D_{1}$ or $D_{2}$. Then $p_{d_{a}d_{b}d_{1}d_{2}}=p_{d_{a}d_{b}\overline{d}_{1}d_{2}}=p_{d_{a}d_{b}d_{1}\overline{d}_{2}}=p_{d_{a}d_{b}\overline{d}_{1}\overline{d}_{2}}=p_{\overline{d}_{a}d_{b}d_{1}d_{2}}=p_{d_{a}\overline{d}_{b}d_{1}d_{2}}=p_{\overline{d}_{a}\overline{d}_{b}d_{1}d_{2}}=p_{\overline{d}_{a}\overline{d}_{b}\overline{d}_{1}d_{2}}=p_{\overline{d}_{a}\overline{d}_{b}d_{1}\overline{d}_{2}}=p_{\overline{d}_{a}\overline{d}_{b}\overline{d}_{1}\overline{d}_{2}}==p_{\overline{d}_{a}d_{b}\overline{d}_{1}\overline{d}_{2}}=p_{d_{a}\overline{d}_{b}\overline{d}_{1}\overline{d}_{2}}=0$,
and $p_{\overline{d}_{a}d_{b}\overline{d}_{1}d_{2}}=p_{\overline{d}_{a}d_{b}d_{1}\overline{d}_{2}}=p_{d_{a}\overline{d}_{b}\overline{d}_{1}d_{2}}=p_{d_{a}\overline{d}_{b}d_{1}\overline{d}_{2}}=\frac{1}{4}$. 
\end{description}
As with the destructive measurements, we have inconsistencies between
the two complementary experimental conditions. This shows that which-path
information creates a context that is different from the one leading
to interference. In other words, the joint probabilities obtained
in the non-destructive measurement are once again incompatible with
the marginals for the interference patterns contained in Case 5. 

In both types of experiments, described by Cases 1 through 8, the
incompatibility of contexts is reflected in the non-existence of a
joint (quasi) negative probability distribution for all possible outcomes.
This reflects the strong contextuality of each setup, interference
or which-path, leading to observables $D_{1}$ and $D_{2}$ that are
contextuality biased. It is interesting at this point to notice that
this would represent, in a trivial way, a case where an experimenter
could choose to observe or not $D_{A}$ or $D_{B}$, and such observation
would change the probabilities in $D_{1}$ and $D_{2}$. Thus, in
a trivial sense, the observation of, say, $D_{A}$ or not could be
used to \emph{signal} another experimenter at $D_{1}$. Though this
is \emph{not} what is usually called signaling in the literature,
as it does not involve any spacelike separations between a transmitter
and a receiver, it does clarify the relationship between the absence
of contextual bias and the no-signalling condition. To distinguish
this, \citet{kofler_condition_2013} coined the term \emph{signaling
in time, }but here we use the term \emph{contextual measurement biases
}suggested by \citet{dzhafarov_contextuality_2014}. See \citet{oas_exploring_2014,dzhafarov_contextuality_2014}
for a somewhat more detailed discussion of this point, including the
relationship between the existence of probability distributions (including
negative ones) and signaling. 

Case 5-8 are equivalent to the spirit of Feynman's discussions about
the double-slit in his 1987 paper, and has been experimentally realized
by \citet{scully_feynmans_1994}. Even though \citet{scully_feynmans_1994}
had access to the outcomes of Case 6, 7, and 8 to infer the joint
probability distribution, they used counterfactual reasoning to compute
a negative probability distribution that was consistent with Case
5. To do so, they had to discard certain measurements from their marginals,
say, by only looking at cases where no detection happened at detector
$B$ in case 7. Let us examine the details of this counterfactual
computations. First one needs to determine what are the actual observable
conditions that constrain the marginal distributions. If we put detectors
on both paths, we observe 
\begin{equation}
P(d_{a}d_{b})=0,\label{eq:p1p2}
\end{equation}
 
\begin{equation}
P(\overline{d_{a}}d_{b})=\frac{1}{2},
\end{equation}
 
\begin{equation}
P(d_{a}\overline{d_{b}})=\frac{1}{2},
\end{equation}
 and 
\begin{equation}
P(\overline{d_{a}}\overline{d_{b}})=0,\label{eq:p1bp2b}
\end{equation}
which corresponds to having only one photon at a time. 

Whenever we observe in detector $D_{1}$ we do not observe in $D_{2}$,
and \emph{vice versa}. Furthermore, since we have a single photon,
we never observe in both detectors or in neither. Finally, interference
requires that we only observe in $D_{1}$. Therefore, 
\begin{equation}
P(d_{1}\overline{d}_{2})=1\label{eq:d1d2}
\end{equation}
 
\begin{equation}
P(\overline{d}_{1}d_{2})=P(d_{1}d_{2})=P(\overline{d}_{1}\overline{d}_{2})=0.\label{eq:d1notd2andothers}
\end{equation}

Now for what Feynman considered the disturbing issue. If we put a
detector in arm $A$ or $B$, from (\ref{eq:p1p2})-(\ref{eq:p1bp2b})
we can ``infer'' that whenever we observe the particle \emph{not
}being in $A$, then the particle must be (probability $1$) in $B$.
But when we block the path, the probabilities are 
\begin{equation}
P(\overline{d_{a}}d_{1}\overline{d}_{2})=P(\overline{d_{a}}\overline{d}_{1}d_{2})=\frac{1}{2},\label{eq:p1detect}
\end{equation}
 and 
\begin{equation}
P(\overline{d_{b}}d_{1}\overline{d}_{2})=P(\overline{d_{b}}\overline{d}_{1}d_{2})=\frac{1}{2}.\label{eq:p2detect}
\end{equation}
 The ``disturbing'' aspect comes from the nonmonotonicity of the
above probabilities. How can $P(\overline{d}_{1}d_{2})=0$, according
to $\eqref{eq:d1notd2andothers}$, while $P(\overline{d_{a}}\overline{d}_{1}d_{2})=\frac{1}{2}$,
from (\ref{eq:p1detect}), given that $\overline{d_{a}}\overline{d}_{1}d_{2}$
is a proper subset of all events where $\overline{d}_{1}d_{2}$?

Of course, this nonmonotonic property cannot be reproduced by Kolmogorov's
axioms. To see this, let $S_{1}$ and $S_{2}$ be two sets in $\mathcal{F}$
such that $S_{1}\subseteq S_{2}$ (as is the case for $S_{1}=\left\{ \omega_{i}\in\Omega|\mathbf{D}_{1}=-1,\mathbf{D}_{2}=1\right\} $
and $S_{2}=\left\{ \omega_{i}\in\Omega|\mathbf{A}=-1,\mathbf{D}_{1}=-1,\mathbf{D}_{2}=1\right\} $
above). Then, we can construct a set $S_{1}'=S_{2}\backslash S_{1}$
such that $S_{1}\cup S_{1}'=S_{2}$ and $S_{1}\cap S_{1}'=\emptyset$.
From K3 we have that $p\left(S_{1}\cup S_{1}'\right)=p\left(S_{1}\right)+p\left(S_{1}'\right)=p\left(S_{2}\right)$,
and from K1 we have at once that $p\left(S_{1}\right)\leq p\left(S_{2}\right)$
if $S_{1}\subseteq S_{2}$, which is clearly violated by the probabilities
above. Notice that in order to prove monotonicity, we had to use the
non-negativity axiom K1. However, since negative probabilities violate
K1, they may be nonmonotonic. For instance, from $P\left(S_{1}\cup S_{1}'\right)=P\left(S_{1}\right)+P\left(S_{1}'\right)=p\left(S_{2}\right)$,
it is possible to have $P\left(S_{1}\right)>P\left(S_{2}\right)$
if $P\left(S_{1}'\right)<0$. 

Before we compute the joint (quasi) negative probabilities from the
assumptions above, let us examine in more detail (\ref{eq:p1detect})
and (\ref{eq:p2detect}). The fact that each add to one corresponds
to a selection of experiments where no detection happens on $D_{A}$
or $D_{B}$. In other words, we are only looking at a subset of all
possible experimental outcomes (essentially, this is equivalent to
a postselection of data). In fact, we can see that (\ref{eq:p1detect})
and (\ref{eq:p2detect}) are distinct from what one observes in Case
2 and Case 3 or Case 6 and Case 7, which, as we pointed out earlier,
are incompatible with Case 1 or Case 5, respectively. In this restricted
data set, the counterfactual reasoning leads to a weaker context-dependency
between variables, allowing for the existence of a joint negative
probability distribution, as we now show. 

From (\ref{eq:p1detect}) and (\ref{eq:p2detect}) we obtain the following
set of linear equations 
\begin{equation}
P\left(\overline{d_{a}}\cdot d_{1}\overline{d}_{2}\right)=P\left(\overline{d_{a}}d_{b}d_{1}\overline{d}_{2}\right)+P\left(\overline{d_{a}}\overline{d_{b}}d_{1}\overline{d}_{2}\right)=\frac{1}{2},\label{eq:system1}
\end{equation}
 
\begin{equation}
P\left(\overline{d_{a}}\cdot\overline{d}_{1}d_{2}\right)=P\left(\overline{d_{a}}d_{b}\overline{d}_{1}d_{2}\right)+P\left(\overline{d_{a}}\overline{d_{b}}\overline{d}_{1}d_{2}\right)=\frac{1}{2},\label{eq:system2}
\end{equation}
 
\begin{equation}
P\left(\cdot\overline{d_{b}}\overline{d}_{1}d_{2}\right)=P\left(d_{a}\overline{d_{b}}\overline{d}_{1}d_{2}\right)+P\left(\overline{d_{a}}\overline{d_{b}}\overline{d}_{1}d_{2}\right)=\frac{1}{2},\label{eq:system2-1}
\end{equation}
and 
\begin{equation}
P\left(\cdot\overline{d_{b}}d_{1}\overline{d}_{2}\right)=P\left(d_{a}\overline{d_{b}}d_{1}\overline{d}_{2}\right)+P\left(\overline{d_{a}}\overline{d_{b}}d_{1}\overline{d}_{2}\right)=\frac{1}{2}.\label{eq:system2-2}
\end{equation}
 From (\ref{eq:d1d2})--(\ref{eq:d1notd2andothers}), we also obtain
that 
\begin{equation}
P\left(\cdot\cdot d_{1}d_{2}\right)=P\left(d_{a}d_{b}d_{1}d_{2}\right)+P\left(\overline{d_{a}}d_{b}d_{1}d_{2}\right)+P\left(d_{a}\overline{d_{b}}d_{1}d_{2}\right)+P\left(\overline{d_{a}}\overline{d_{b}}d_{1}d_{2}\right)=0,\label{eq:system3}
\end{equation}
 
\begin{equation}
P\left(\cdot\cdot d_{1}\overline{d}_{2}\right)=P\left(d_{a}d_{b}d_{1}\overline{d}_{2}\right)+P\left(\overline{d_{a}}d_{b}d_{1}\overline{d}_{2}\right)+P\left(d_{a}\overline{d_{b}}d_{1}\overline{d}_{2}\right)+P\left(\overline{d_{a}}\overline{d_{b}}d_{1}\overline{d}_{2}\right)=1,\label{eq:system4}
\end{equation}
 
\begin{equation}
P\left(\cdot\cdot\overline{d}_{1}d_{2}\right)=P\left(d_{a}d_{b}\overline{d}_{1}d_{2}\right)+P\left(\overline{d_{a}}d_{b}\overline{d}_{1}d_{2}\right)+P\left(d_{a}\overline{d_{b}}\overline{d}_{1}d_{2}\right)+P\left(\overline{d_{a}}\overline{d_{b}}\overline{d}_{1}d_{2}\right)=0,\label{eq:system5}
\end{equation}
 
\begin{equation}
P\left(\cdot\cdot\overline{d}_{1}\overline{d}_{2}\right)=P\left(d_{a}d_{b}\overline{d}_{1}\overline{d}_{2}\right)+P\left(\overline{d_{a}}d_{b}\overline{d}_{1}\overline{d}_{2}\right)+P\left(d_{a}\overline{d_{b}}\overline{d}_{1}\overline{d}_{2}\right)+P\left(\overline{d_{a}}\overline{d_{b}}\overline{d}_{1}\overline{d}_{2}\right)=0.\label{eq:system6}
\end{equation}
 Finally, (\ref{eq:p1p2})--(\ref{eq:p1bp2b}) yields 
\begin{equation}
P\left(d_{a}d_{b}\cdot\cdot\right)=P\left(d_{a}d_{b}d_{1}d_{2}\right)+P\left(d_{a}d_{b}\overline{d}_{1}d_{2}\right)+P\left(d_{a}d_{b}d_{1}\overline{d}_{2}\right)+P\left(d_{a}d_{b}\overline{d}_{1}\overline{d}_{2}\right)=0,\label{eq:system7}
\end{equation}
 
\begin{equation}
P\left(d_{a}\overline{d_{b}}\cdot\cdot\right)=P\left(d_{a}\overline{d_{b}}d_{1}d_{2}\right)+P\left(d_{a}\overline{d_{b}}\overline{d}_{1}d_{2}\right)+P\left(d_{a}\overline{d_{b}}d_{1}\overline{d}_{2}\right)+P\left(d_{a}\overline{d_{b}}\overline{d}_{1}\overline{d}_{2}\right)=\frac{1}{2},\label{eq:system8}
\end{equation}
 
\begin{equation}
P\left(\overline{d_{a}}d_{b}\cdot\cdot\right)=P\left(\overline{d_{a}}d_{b}d_{1}d_{2}\right)+P\left(\overline{d_{a}}d_{b}\overline{d}_{1}d_{2}\right)+P\left(\overline{d_{a}}d_{b}d_{1}\overline{d}_{2}\right)+P\left(\overline{d_{a}}d_{b}\overline{d}_{1}\overline{d}_{2}\right)=\frac{1}{2},\label{eq:system9}
\end{equation}
 
\begin{equation}
P\left(\overline{d_{a}}\overline{d_{b}}\cdot\cdot\right)=P\left(\overline{d_{a}}\overline{d_{b}}d_{1}d_{2}\right)+P\left(\overline{d_{a}}\overline{d_{b}}\overline{d}_{1}d_{2}\right)+P\left(\overline{d_{a}}\overline{d_{b}}d_{1}\overline{d}_{2}\right)+P\left(\overline{d_{a}}\overline{d_{b}}\overline{d}_{1}\overline{d}_{2}\right)=0.\label{eq:system10}
\end{equation}
 As a last condition, from axiom N2, all probabilities of elementary
events must add to one, i.e., 
\begin{eqnarray}
P\left(d_{a}d_{b}d_{1}d_{2}\right)+P\left(d_{a}d_{b}d_{1}\overline{d}_{2}\right)+P\left(d_{a}d_{b}\overline{d}_{1}d_{2}\right)+P\left(d_{a}d_{b}\overline{d}_{1}\overline{d}_{2}\right)+\label{eq:system11}\\
P\left(d_{a}\overline{d_{b}}d_{1}d_{2}\right)+P\left(a\overline{d_{b}}d_{1}\overline{d}_{2}\right)+P\left(a\overline{d_{b}}\overline{d}_{1}d_{2}\right)+P\left(a\overline{d_{b}}\overline{d}_{1}\overline{d}_{2}\right)+\nonumber \\
P\left(\overline{d_{a}}d_{b}d_{1}d_{2}\right)+P\left(\overline{d_{a}}d_{b}d_{1}\overline{d}_{2}\right)+P\left(\overline{d_{a}}d_{b}d_{1}d_{2}\right)+P\left(\overline{d_{a}}d_{b}\overline{d}_{1}\overline{d}_{2}\right)+\nonumber \\
P\left(\overline{d_{a}}\overline{d_{b}}d_{1}d_{2}\right)+P\left(\overline{d_{a}}\overline{d_{b}}d_{1}\overline{d}_{2}\right)+P\left(\overline{d_{a}}\overline{d_{b}}\overline{d}_{1}d_{2}\right)+P\left(\overline{d_{a}}\overline{d_{b}}\overline{d}_{1}\overline{d}_{2}\right) & = & 1\nonumber 
\end{eqnarray}
The general solution to the underdetermined (and not independent)
system of equations (\ref{eq:system1})--(\ref{eq:system11}) is 
\begin{equation}
\begin{array}[t]{cc}
P\left(d_{a}d_{b}d_{1}d_{2}\right)=\alpha, & P\left(d_{a}d_{b}d_{1}\overline{d}_{2}\right)=\theta+\frac{1}{2}\left(\delta-\gamma+\beta-\alpha\right),\\
P\left(d_{a}d_{b}\overline{d}_{1}d_{2}\right)=-\frac{1}{2}-\theta, & P\left(d_{a}d_{b}\overline{d}_{1}\overline{d}_{2}\right)=\frac{1}{2}+\frac{1}{2}(-\delta+\gamma-\beta-\alpha),\\
P\left(d_{a}\overline{d_{b}}d_{1}d_{2}\right)=\frac{1}{2}(-\delta-\gamma+\beta-\alpha), & P\left(d_{a}\overline{d_{b}}d_{1}\overline{d}_{2}\right)=\frac{1}{2}-\theta+\frac{1}{2}(-\delta+\gamma-\beta+\alpha)\\
P\left(d_{a}\overline{d_{b}}\overline{d}_{1}d_{2}\right)=\theta, & P\left(d_{a}\overline{d_{b}}\overline{d}_{1}\overline{d}_{2}\right)=\delta,\\
P\left(\overline{d_{a}}bd_{b}d_{1}d_{2}\right)=\frac{1}{2}(\delta-\gamma-\beta-\alpha), & P\left(\overline{d_{a}}d_{b}d_{1}\overline{d}_{2}\right)=\frac{1}{2}-\theta+\frac{1}{2}(-\delta+\gamma-\beta+\alpha),\\
P\left(\overline{d_{a}}d_{b}\overline{d}_{1}d_{2}\right)=\theta, & P\left(\overline{d_{a}}d_{b}\overline{d}_{1}\overline{d}_{2}\right)=\beta,\\
P\left(\overline{d_{a}}\overline{d_{b}}d_{1}d_{2}\right)=\gamma, & P\left(\overline{d_{a}}\overline{d_{b}}d_{1}\overline{d}_{2}\right)=\theta+\frac{1}{2}(\delta-\gamma+\beta-\alpha),\\
P\left(\overline{d_{a}}\overline{d_{b}}\overline{d}_{1}d_{2}\right)=\frac{1}{2}-\theta, & P\left(\overline{d_{a}}\overline{d_{b}}\overline{d}_{1}\overline{d}_{2}\right)=-\frac{1}{2}+\frac{1}{2}(-\delta-\gamma-\beta+\alpha),
\end{array}\label{eq:general-solution}
\end{equation}
where $\alpha$, $\beta$, $\gamma$, $\delta$, and $\theta$ are
arbitrary constants. It is clear from (\ref{eq:general-solution})
that no nonnegative solution exists for (\ref{eq:system1})--(\ref{eq:system11}).
Furthermore, because the system is underdetermined, there are an infinite
number of solutions that satisfy (\ref{eq:system1})--(\ref{eq:system11}).
To find the negative probabilities, though, we need to minimize the
L1 norm, $M^{*}$. Doing so for (\ref{eq:general-solution}) is straightforward
but tedious, and we can show that such minimum happens when $0\leq\alpha\leq\frac{1}{2}$,
$\beta=0$, $\delta=0$, $\theta=0$, and $\alpha=-\gamma$. This
gives us the general solution minimizing $M^{*}$ as 
\begin{equation}
\begin{array}[t]{cc}
P\left(d_{a}d_{b}d_{1}d_{2}\right)=\alpha, & P\left(d_{a}d_{b}d_{1}\overline{d}_{2}\right)=0,\\
P\left(d_{a}d_{b}\overline{d}_{1}d_{2}\right)=-\frac{1}{2}, & P\left(d_{a}d_{b}\overline{d}_{1}\overline{d}_{2}\right)=\frac{1}{2}-\alpha,\\
P\left(d_{a}\overline{d_{b}}d_{1}d_{2}\right)=0, & P\left(d_{a}\overline{d_{b}}d_{1}\overline{d}_{2}\right)=\frac{1}{2},\\
P\left(d_{a}\overline{d_{b}}\overline{d}_{1}d_{2}\right)=0, & P\left(d_{a}\overline{d_{b}}\overline{d}_{1}\overline{d}_{2}\right)=0,\\
P\left(\overline{d_{a}}d_{b}d_{1}d_{2}\right)=0, & P\left(\overline{d_{a}}d_{b}d_{1}\overline{d}_{2}\right)=\frac{1}{2},\\
P\left(\overline{d_{a}}d_{b}\overline{d}_{1}d_{2}\right)=0, & P\left(\overline{d_{a}}d_{b}\overline{d}_{1}\overline{d}_{2}\right)=0,\\
P\left(\overline{d_{a}}\overline{d_{b}}d_{1}d_{2}\right)=-\alpha, & P\left(\overline{d_{a}}\overline{d_{b}}d_{1}\overline{d}_{2}\right)=0,\\
P\left(\overline{d_{a}}\overline{d_{b}}\overline{d}_{1}d_{2}\right)=\frac{1}{2}, & P\left(\overline{d_{a}}\overline{bd_{b}}\overline{d}_{1}\overline{d}_{2}\right)=-\frac{1}{2}+\alpha,
\end{array}\label{eq:general-solution-minimizing}
\end{equation}
$0\leq\alpha\leq\frac{1}{2}$, clearly showing that $M^{*}=3$. We
should notice that this value of $M^{*}$ is greater than the $M^{*}=2$
for the Bell-EPR case \citep{oas_exploring_2014}, perhaps already
suggesting that the double-slit is more contextual (see \citet{de_barros_unifying_2014}
for a discussion of $M^{*}$ as a measure of contextuality). This
stronger contextuality probably comes from the use of triple moments
in the Mach-Zehnder as opposed to only pairwise two-moments in the
case of the standard Bell-EPR setup. 

Now that we have a negative probability distribution, we can use it
to compute conditional probabilities based on the previous counterfactual
assumptions. For instance, a standard question is this: if a photon
is detected on $D_{1}$, what is the probability that this photon
went through $A$ and $B$? Using 
\begin{eqnarray*}
P(d_{a}|d_{1}) & = & \frac{1}{N}\left[P\left(d_{a}d_{b}d_{1}d_{2}\right)+P\left(d_{a}\overline{d_{b}}d_{1}d_{2}\right)+P\left(d_{a}d_{b}d_{1}\overline{d}_{2}\right)+P\left(d_{a}\overline{d_{a}}d_{1}\overline{d}_{2}\right)\right],
\end{eqnarray*}
where 
\begin{eqnarray*}
N & = & P\left(d_{a}d_{b}d_{1}d_{2}\right)+P\left(\overline{d_{a}}d_{b}d_{1}d_{2}\right)+P\left(d_{a}\overline{d_{b}}d_{1}d_{2}\right)+P\left(d_{a}d_{b}d_{1}\overline{d}_{2}\right)\\
 &  & +P\left(\overline{d_{a}}\overline{d_{b}}d_{1}d_{2}\right)+P\left(\overline{d_{a}}d_{b}d_{1}\overline{d}_{2}\right)+P\left(d_{a}\overline{d_{b}}d_{1}\overline{d}_{2}\right)+P\left(\overline{d_{a}}\overline{d_{b}}d_{1}\overline{d}_{2}\right),
\end{eqnarray*}
and (\ref{eq:general-solution-minimizing}) we have that
\begin{eqnarray*}
P(d_{a}|d_{1}) & = & \frac{1}{2}+\alpha,
\end{eqnarray*}
and 
\[
\frac{1}{2}\leq P\left(d_{a}|d_{1}\right)\leq1.
\]
Similarly,
\begin{eqnarray*}
P(d_{b}|d_{1}) & = & \frac{1}{N}\left[P\left(d_{a}d_{b}d_{1}d_{2}\right)+P\left(\overline{d_{a}}d_{b}d_{1}d_{2}\right)+P\left(d_{a}d_{b}d_{1}\overline{d}_{2}\right)+P\left(\overline{d_{a}}d_{b}d_{1}\overline{d}_{2}\right)\right],
\end{eqnarray*}
where 
\begin{eqnarray*}
N & = & P\left(d_{a}d_{b}dd_{2}\right)+P\left(d_{a}\overline{d_{b}}d_{1}d_{2}\right)+P\left(\overline{d_{a}}d_{b}d_{1}d_{2}\right)+P\left(d_{a}d_{b}d_{1}\overline{d}_{2}\right)\\
 &  & +P\left(\overline{d_{a}}\overline{d_{b}}d_{1}d_{2}\right)+P\left(d_{a}\overline{d_{b}}d_{1}\overline{d}_{2}\right)+P\left(\overline{d_{a}}d_{b}d\overline{d}_{2}\right)+P\left(\overline{d_{a}}\overline{d_{b}}d_{1}\overline{d}_{2}\right),
\end{eqnarray*}
and we get
\begin{eqnarray*}
P(d_{b}|d_{1}) & = & \frac{1}{2}+\alpha,
\end{eqnarray*}
the same value we got for the conditional $P(d_{b}|d_{1})$. For $d_{2}$,
the conditional probability is not defined, as the probability for
$d_{2}$ from the joint is zero. However, if we set the interferometer
such that the probability of $d_{2}$ is not zero, but close to it,
then $P(b|d_{2})$ can be shown to approach $P(d_{b}|d_{2})=-\frac{1}{2}+\alpha$.
If that is the case, it is reasonable to assume that $\alpha=1/2$,
such that we do not have negative probabilities for $d_{b}$ (conditioned
on $d_{2}$). If we do so, we reach the interesting conclusion that
both $d_{b}$ and $d_{a}$ have probability $1$ given an observation
on $d_{1}$. In other words, if we use the counterfactual reasoning
from negative probabilities, we reach the conclusion, as Feynman often
said, that the particle goes through \emph{both }paths simultaneously. 

We now end this section with a discussion of some well-known uses
of counterfactual reasoning in quantum mechanics and their relationship
to our discussion above. First, it is worth mentioning that the famous
Leggett and Garg \citeyearpar{leggett_quantum_1985} setup can be
thought of as similar to our double-slit experiment \citep{kofler_condition_2013}.
To see this, we recall that in Leggett and Garg (LG), measurements
in three distinct times can coded by three $\pm1$-valued random variables,
say $\mathbf{X}$, $\mathbf{Y}$, and $\mathbf{Z}$ \citep{bacciagaluppi_leggett-garg_2014,dzhafarov_contextuality_2014,de_barros_unifying_2014,dzhafarov_generalizing_2014}.
In an analogy with the double-slit, and following \citet{kofler_condition_2013},
we can think of $\mathbf{X}$ as a measurement of position before
$BS_{1}$, $\mathbf{Y}$ as a measurement of which path (say, with
$\mathbf{Y}=1$ corresponding to $A$ and $\mathbf{Y}=-1$ to $B$),
and $\mathbf{Z}$ corresponding to a detection in either $D_{1}$
(for $\mathbf{Z}=1$) or $D_{2}$ (\textbf{$\mathbf{Z}=-1$}). This
case would correspond to contextual bias, and would not include counterfactual
reasoning. However, in the original LG paper, counterfactual reasoning
happens by not measuring $\mathbf{Y}$, but instead making inferences
about $\mathbf{Y}$ from an absence of detection in one of the paths.
As we mentioned above, such contextual bias is not surprising, as
the effect of measuring could be thought as interfering with the experimental
conditions themselves. This is similar to what happens in our analysis
above. 

Something analogous happens with the argument given by \citet{scully_feynmans_1994}.
In his paper, he talks about negative probabilities, and shows that
they lead to the interference between two possible modes. However,
it is easy to see that the negative probabilities so obtained are
only existent because of the same type of counterfactual reasoning
shown above. That should be clear by the fact that, in their experiment,
the interference pattern is existent, and therefore we have observational
contextual bias, similar to the Leggett-Garg setup.

\section{Final Remarks}

In this paper we presented a proposed theory of negative probabilities
that could be used to describe non-monotonic reasoning. Such theory
was shown to be equivalent, in the case when a proper probability
distribution exists, to the standard Kolmogorov probability, as the
requirement of minimizing the total probability mass leads to a Kolmogorovian
distribution. Furthermore, in cases where no proper joint probability
exists, the minimization of the total negative mass is simply a requirement
that our quasi-probability distribution is as close to a proper one
as possible. Such minimization is similar to the requirement with
upper probabilities of minimizing the total sum of probabilities,
which may exceed one. 

We did not attempt to interpret negative probabilities, but instead
took the approach that they constitute a bookkeeping tool that meets
a minimum rationality criteria of minimization of the L1 probability
norm. This criteria is perhaps not without practical consequences.
For instance, in \citet{de_barros_decision_2014}, we showed that
in certain cases where the pairwise correlations lead to contradictions,
this minimization results in constraints to the triple moments. In
addition, in \citet{oas_exploring_2014}, it was shown that the minimized
L1 norm is equivalent to the CHSH parameter, $S$, as used in measures
of non-locality \citep{cirelson_quantum_1980}, specifically $M^{*}=S/2$.
Finally, the L1 norm can be thought of as a measure of contextuality
for random variables, and is closely related to other measures of
contextuality, at least for three and four random-variables, and for
more variables it suggests possible different classifications for
contextuality \citep{de_barros_unifying_2014}. 

Though the double-slit experiment is the archetypical in discussions
of how quantum mechanics leads to violation of the laws of probability
or logic \citep{fano_two-slits_2014}, it is not the simplest and
most accessible example that contains the key conceptual elements
relevant to the subject. In fact, the Mach-Zehnder interferometer,
as we showed above, presents the same characteristics as the double-slit
experiment that are relevant to conceptual discussions of quantum
mechanics, without the complications associated with the details of
continuous interference patterns present in the double-slit. For that
reason, in our discussion of the double-slit experiment in terms of
negative probabilities, we resorted to the simpler case of the Mach-Zehnder
interferometer. As we saw, the Mach-Zehnder interferometer allows
us to talk about the features of double-slit experiment in terms of
discrete random variables, which tremendously reduce the mathematical
complexity without loosing any conceptual generality.

As we showed, the two possible setups for the Mach-Zehnder interferometer,
one with which-path information and another with interference, present
contextual biases. This has, as a consequence, the non-existence of
a joint negative (quasi) probability distribution consistent with
all observations of the two Mach-Zehnder interferometer setups. This
clearly corresponds to two different experimental contexts, and not
only does a joint probability consistent with both contexts not exist,
but no negative joint distribution exists either. This type of system,
where contextuality comes from contextual biases, exhibit what one
could think of as stronger contextuality than other systems, such
as EPR. 

This stronger contextuality may be what is reflected in the large
values of $M^{*}$ for the two slits, but such a connection has not
been studied in detail. In fact, it is interesting to notice that
the large values of $M^{*}$ is associated to a set of observables
that do not provide a complete picture of the experimental conditions,
as it relies on counterfactuals. Perhaps there is a connection between
large $M^{*}$ for a restricted set of observables and contextual
biases for an extended set, which could provide an interesting criteria
for contextual biases. Notice that, as mentioned in Section \ref{sec:The-Mach-Zehnder-Interferometer},
contextual biases are equivalent to the violation of the no-signaling
condition in multipartite systems.

Finally, we would like to comment on Feynman's \citeyearpar{feynman_negative_1987}
discussion of the double-slit experiment. In this paper, he argues
that the non-monotonic character of quantum probabilities could be
represented by non-observable negative probabilities. He then goes
on and constructs (in a very informal way) a possible negative probability
that could explain the outcomes of the experiment. However, as we
pointed out above, negative probabilities consistent with the outcomes
of the double-slit experiment are impossible, unless we make use of
certain specific counterfactual reasoning. It is interesting to note
that in his actual experimental realization of Feynman's double-slit
experiment, Scully et al. \citeyearpar{scully_feynmans_1994} construct
a negative probability distribution, and their probabilities rely
exactly on the type of counterfactual reasoning used above. Were they
to try and construct a negative probability from the full range of
experimental data, they would not be able to do so. A similar case
is present in the LG experiment also discussed above.

\paragraph*{Acknowledgments. }

This paper stems from a 2011 seminar led by the three authors, where
a full review and exploration of negative probabilities was undertaken,
and we profited from exchanges with the seminar participants, in particular
Tom Ryckman, Claudio Carvalhaes, Michael Heaney, Michael Friedman,
Hyungrok Kim, and Niklas Damiris. It was Pat Suppes\textquoteright s
goal to establish a rigorous, consistent theory of negative probability,
and this work is in line with that unfinished goal. We also benefited
tremendously from collaborations and discussions with Ehtibar Dzhafarov
and Janne Kujala. We also thank Guido Bacciagaluppi for discussions
and for pointing us to the Kofler and Brukner reference. 

\bibliographystyle{apa}
\bibliography{Quantum}

\end{document}